\documentclass[research]{fcs}
\usepackage[T1]{fontenc}
\usepackage{bbm}
\usepackage{amsmath,amsfonts}
\usepackage{algorithmic}
\usepackage{algorithm}
\usepackage{array}
\usepackage{textcomp}
\usepackage{stfloats}
\usepackage{float}
\usepackage{url}
\usepackage{verbatim}
\usepackage{graphicx}

\usepackage{bbding}
\usepackage{threeparttable}
\usepackage{subcaption}
\usepackage{hyperref}
\usepackage{cleveref}
\usepackage{makecell}
\usepackage{svg}
\Crefname{figure}{fig}{figures}
\Crefname{figure}{Fig}{Figures}
\usepackage{cite}
\usepackage{epsfig,endnotes,tikz,pgfplots}
\usepackage{xcolor}

\newcommand{\blue}{\textcolor{black}}

\title{Registered Attribute-Based Encryption with Reliable Outsourced Decryption Based on Blockchain}
\shorttitle{Registered ABE with Reliable Outsourced Decryption}

\author[1]{Dongliang~CAI}
\author[2,4]{Liang~ZHANG}
\author[1]{Borui~CHEN}
\author[1,3,+]{Haibin~KAN}
\address[1]{College of Computer Science and Artificial Intelligence, Fudan University, Shanghai 200433, China}
\address[2]{Department of Industrial Engineering and Decision Analytics, Hong Kong University of Science and Technology, Hong Kong}
\address[3]{Shanghai Institute for Mathematics and Interdisciplinary Sciences, Shanghai 200433, China}
\address[4]{School of Cyberspace Security (School of Cryptology), Hainan University, Haikou, 570228, China}

\corremail{hbkan@fudan.edu.cn}
\fcssetup{
  received = {month dd, yyyy},
  accepted = {month dd, yyyy},
}

\begin{abstract}
Decentralized data sovereignty and secure data exchange are regarded as foundational pillars of the new era. Attribute-based encryption (ABE) is a promising solution that enables fine-grained access control in data sharing. Recently, Hohenberger et al. (Eurocrypt 2023) introduced registered ABE (RABE) to eliminate trusted authority and gain decentralization. Users generate their own public and secret keys and then register their keys and attributes with a transparent key curator. However, RABE still suffers from heavy decryption overhead. A natural approach to address this issue is to outsource decryption to a decryption cloud server (DCS). In this work, we propose the first auditable RABE scheme with reliable outsourced decryption (ORABE) based on blockchain. First, we achieve verifiability of transform ciphertext via a verifiable tag mechanism. Then, the exemptibility, which ensures that the DCS escapes false accusations, is guaranteed by zero knowledge fraud proof under the optimistic assumption. Additionally, our system achieves fairness and auditability to protect the interests of all parties through blockchain. Finally, we give concrete security and theoretical analysis and evaluate our scheme on Ethereum to demonstrate feasibility and efficiency.
\end{abstract}

\begin{document}

\section{Introduction}

As the digital world evolves, the metaverse is gaining traction\cite{gan2023web}, creating new opportunities and challenges for data sharing. The metaverse requires a decentralized and user-centric data sharing framework. Ciphertext-policy attribute-based encryption (CP-ABE)\cite{bethencourt2007ciphertext} provides a promising primitive to achieve fine-grained access control for data sharing by embedding access structure directly into the encrypted data and associating attributes with decryption keys.

CP-ABE relies on a trusted authority that maintains a master secret key to issue decryption keys\cite{bethencourt2007ciphertext}, limiting its wide adoption in a decentralized environment. If the central authority is compromised, the adversary would be able to decrypt all ciphertexts within the system. The proposal of decentralized ABE still relies on multiple authorities\cite{lewko2011decentralizing}. Therefore, Hohenberger et al.\cite{hohenberger2023registered} proposed a registered ABE (RABE) scheme without an authority. In RABE, users generate their secret and public keys and register their keys and attributes with a key curator. The curator is transparent, keeps no secrets, and its malicious behavior can always be detected. Subsequently, several constructions of RABE were proposed\cite{garg2024reducing,FWW2023use,ZZGQ2023registered,attrapadung2024modular,weng2025subpolicy,li2025revocablerabe} to enhance the original scheme. The pairing-based constructions are more efficient, but many of them \cite{hohenberger2023registered,ZZGQ2023registered,attrapadung2024modular} have a large and impractical common reference string (CRS). Garg et al.\cite{garg2024reducing} reduced the size of CRS with progression-free sets, which makes the scheme more practical. \blue{Considering real-world applications, Weng et al. \cite{weng2025subpolicy} reduced the decryption overhead for ciphertexts with the same sub-policy in RABE, while Li et al. \cite{li2025revocablerabe} proposed a revocable RABE scheme to enhance flexibility.}

However, ABE schemes\cite{bethencourt2007ciphertext,goyal2006attribute,waters2011ciphertext,hohenberger2023registered, garg2024reducing} suffer from high decryption overhead, which poses a significant challenge for lightweight devices. To address this issue, Green et al. \cite{green2011outsourcing} proposed an outsourced decryption scheme for attribute-based encryption (OABE). OABE reduces the computational overhead for users without revealing information about the original data. In their scheme, the decryption cloud server (DCS) transforms the ciphertext into an ElGamal ciphertext \cite{elgamal1985public}, which the user then decrypts to retrieve the plaintext. Since the DCS may not perform the computation honestly, many verifiable OABE schemes have been proposed\cite{lai2013attribute,li2013securely,qin2015attribute,lin2015revisiting,miao2023verifiable}. \blue{Xu et al. \cite{xu2015circuit} and Li et al. \cite{li2017full} extended verifiability to ensure that both authorized and unauthorized users can verify the correctness of the outsourced decryption, which is called full verifiability in \cite{li2017full}.} Nevertheless, many previous schemes fail to guarantee exemptibility and fairness. The exemptibility property ensures that the DCS is not falsely accused if it returns a correct transform ciphertext. The fairness property ensures that the DCS gets paid if and only if it returns a correct result. To tackle this issue, Cui et al. \cite{cui2020pay} introduced a payable OABE scheme based on blockchain \cite{nakamoto2008bitcoin}. 
Due to the lack of support for pairing computation in Ethereum\cite{wood2014ethereum}, the OABE scheme proposed by Ge et al. \cite{ge2023attribute} cannot be directly integrated with Ethereum platform.

To the best of our knowledge, no outsourced decryption scheme has been proposed for RABE. Moreover, there is a lack of a reliable and efficient outsourced decryption scheme for ABE.
This gap motivates our work. In this paper, we develop the first reliable and efficient outsourced decryption solution for the RABE scheme proposed by Garg et al. \cite{garg2024reducing} based on Ethereum blockchain.
Our contributions of this paper are fourfold:
\begin{itemize}
    \item We propose the first auditable RABE scheme with reliable outsourced decryption (ORABE) based on blockchain, which reduces the user's decryption overhead to only one group exponentiation operation without \blue{requiring} pairing operations.
    \item We achieve verifiability of transform ciphertext by a verifiable tag mechanism\cite{qin2015attribute} and ensure exemptibility via non-interactive zero knowledge (NIZK) proof under the optimistic assumption. Moreover, auditability and fairness are ensured based on decentralized blockchain and the payment can be processed by cryptocurrency.
    \item We propose a decentralized and auditable data sharing framework without a trusted authority. The reliable outsourced decryption makes it friendly for user-centric lightweight devices.
    \item We give a concrete security analysis of the proposed scheme and implement it on Ethereum to evaluate its feasibility and performance.
\end{itemize}

\section{Related Work}
\subsection{Outsourced Decryption of Attribute-Based Encryption} ABE was first introduced as fuzzy identity-based encryption by Sahai and Waters in \cite{sahai2005fuzzy}. In CP-ABE, the user can decrypt only if the set of attributes satisfies the access policy of ciphertext. However, the high decryption overhead in ABE greatly hinders its real-world adoption. To mitigate this issue, Green et al. \cite{green2011outsourcing} proposed an OABE scheme, which outsources the majority of the decryption to a third party without revealing the plaintext. Consequently, Lai et al. \cite{lai2013attribute} extended the OABE model to achieve verifiability of transform ciphertext. Their approach appends a ciphertext of a random message along with a tag to the original ciphertext, which \blue{doubles} both the ciphertext size and the transformation cost. Qin et al.\cite{qin2015attribute} reduced the ciphertext size and achieved faster decryption than the scheme proposed in \cite{lai2013attribute}. Lin et al. \cite{lin2015revisiting} proposed a more efficient and generic construction of ABE with verifiable outsourced decryption based on an attribute-based key encapsulation mechanism and a symmetric-key encryption scheme. \blue{Li et al. \cite{li2017full} extended verifiability to full verifiability, ensuring that both authorized and unauthorized users can verify, which encrypts the plaintext with two access policies for authorized and unauthorized users, respectively.}
Furthermore, Li et al. \cite{li2013securely} proposed an ABE scheme with outsourced encryption and decryption, but it only supports the threshold access policy. Following their work, Ma et al. \cite{ma2015verifiable} extended it to accommodate any monotonic access policy and established a stronger security framework compared to \cite{lai2013attribute}. However, the exemptibility they claimed is not achieved since the user may reveal an incorrect private key and accuse the DCS falsely. None of the aforementioned schemes can achieve both exemptibility and fairness. To tackle this issue, Cui et al. \cite{cui2020pay} introduced a payable OABE scheme that leverages blockchain technology to achieve exemptibility and fairness. However, it still introduces redundant information. Ge et al. \cite{ge2023attribute} eliminated redundant information and mitigated the heavy computation problem of smart contracts through decomposition, but their solution still suffers from high gas consumption and performance limitation. Furthermore, due to the lack of support for pairing computation in native Ethereum\cite{wood2014ethereum}, this scheme is incompatible with Ethereum platform.

\subsection{Registered Attribute-Based Encryption} RABE\cite{hohenberger2023registered} was introduced recently as a significant extension of ABE, addressing the key escrow problem by replacing the trusted authority that maintains a master secret key to issue attribute keys in traditional ABE with a transparent key curator. In RABE system, users select their public and secret keys independently, and the key curator aggregates users' key-attribute pairs into a master public key. A user encrypts a message with an access policy using the master public key to generate a ciphertext. Any user whose attributes satisfy the access policy in the ciphertext can decrypt it using a helper decryption key and their own secret key. Hohenberger et al. \cite{hohenberger2023registered} introduced the concept of RABE and provided a construction based on composite order group. Subsequently, several constructions of RABE \cite{garg2024reducing,DP2023registration,datta2025registered,FWW2023use,ZZGQ2023registered,FFM2023registered,attrapadung2024modular} were proposed to enhance the original scheme. Some constructions\cite{garg2024reducing,DP2023registration,datta2025registered,FWW2023use,FFM2023registered} require non-black-box use of cryptography, whereas pairing-based constructions\cite{garg2024reducing,FFM2023registered,ZZGQ2023registered,attrapadung2024modular} do not. The pairing-based constructions are more efficient, but many of them \cite{hohenberger2023registered,FFM2023registered,ZZGQ2023registered,attrapadung2024modular} have a large CRS size. Garg et al.\cite{garg2024reducing} reduced it to $L^{1+o(1)}$ ($L$ denotes the number of users) with progression-free sets, thereby making the scheme more practical. \blue{Many works explored the improvements of RABE in practical applications recently. Weng et al. \cite{weng2025subpolicy} presented an efficient RABE scheme for ciphertexts with the same sub-policy in cloud storage. The experimental result shows decryption efficiency is significantly improved with a minimal storage trade-off. Motivated by the revocable ABE scheme\cite{chen2023efficient}, Li et al.\cite{li2025revocablerabe} proposed a revocable RABE scheme with user deregistration to improve flexibility.} However, existing pairing-based constructions still suffer from high decryption overhead, which is unfriendly for user-centric lightweight devices. Therefore, we build upon construction\cite{garg2024reducing} and develop a reliable outsourced decryption scheme aimed at minimizing the decryption burden for users.

\section{Preliminaries}
\label{sec:Preliminaries}
\subsection{Bilinear Map}
Let $\mathbb{G}, \mathbb{G}_T$ be two cyclic groups of prime order $p$, $e$ be the bilinear map $\mathbb{G} \times \mathbb{G} \rightarrow \mathbb{G}_T$, and $g$ be the generator of $\mathbb{G}$. The bilinear map has three properties:
\begin{itemize}
    \item \textit{Bilinearity}: $e(g^a,g^b) = e(g,g)^{ab}$ holds for all $a,b \in \mathbb{Z}_{p}^{*}$.
    \item \textit{Non-degeneracy}: $e(g,g) \neq 1$.
    \item \textit{Computability}: $e(g,g)$ can be computed in polynomial time.
\end{itemize}

\subsection{Outsourced Decryption CP-ABE}

\begin{definition}
(\textbf{Outsourced Decryption CP-ABE}\cite{green2011outsourcing}).
An outsourced decryption CP-ABE scheme is a tuple of efficient algorithms $\prod_{OABE} =$ $(\textbf{Setup}, \textbf{KeyGen}, \textbf{Encrypt}, \textbf{TKGen}, \textbf{Transform},$ $\textbf{Decrypt}_{user})$ as follows.
\begin{itemize}
    \item $\textbf{Setup}(\lambda, \mathcal{U}) \rightarrow PK,MSK$. The algorithm takes a security parameter $\lambda$ and a universe description $\mathcal{U}$ as inputs. It outputs a public key \blue{$PK$} and a master key \blue{$MSK$}.
    \item $\textbf{KeyGen}(MSK,S) \rightarrow SK$. The algorithm takes the master key \blue{$MSK$} and a set of attributes \blue{$S$} as inputs. It outputs a secret key \blue{$SK$}.
    \item $\textbf{Encrypt}(m, PK, (M, \rho)) \rightarrow CT$. The algorithm takes a message \blue{$m$}, the public key \blue{$PK$}, and a linear secret sharing scheme (LSSS) access structure $(M, \rho)$ as inputs. It outputs a ciphertext \blue{$CT$}.
    \item $\textbf{TKGen}(SK) \rightarrow TK, z$.  The algorithm takes the secret key \blue{$SK$} as input. It chooses a random value $z\in \mathbb{Z}_{p}^{*}$ and outputs a transform key \blue{$TK$} and a retrieve key \blue{$z$}.
    \item $\textbf{Transform}(CT, TK)\rightarrow CT'$. The algorithm takes the transform key \blue{$TK$} and the ciphertext \blue{$CT$} as inputs. If the transform key does not satisfy the access structure, the algorithm outputs $\perp$, else it outputs a transform ciphertext \blue{$CT'$}. 
    \item $\textbf{Decrypt}_{user}(CT', z) \rightarrow m$. The algorithm takes the transform ciphertext \blue{$CT'$} and the retrieve key \blue{$z$} as inputs. It outputs a message \blue{$m$}.
\end{itemize}
\end{definition}

\subsection{Registered Attribute-Based Encryption}
\begin{definition}
    (\textbf{Registered Attribute-Based Encryption}\cite{hohenberger2023registered}).
    A RABE scheme with security parameter $\lambda$, attribute universe $\mathcal{U}$, and policy space $\mathcal{P}$ is a tuple of algorithms $\prod_{RABE} =(\textbf{Setup},\textbf{KeyGen},\textbf{RegPK},$
    $\textbf{Encrypt},\textbf{Update},\textbf{Decrypt}):$
    \begin{itemize}
        \item $\textbf{Setup}(1^{\lambda}, 1^{|\mathcal{U}|}, 1^{L}) \rightarrow crs$. The algorithm takes the security parameter $\lambda$, the attribute universe $\mathcal{U}$, and the number of users $L$ as inputs and generates a common reference string \blue{$crs$} as output.
        \item $\textbf{KeyGen}(crs, aux) \rightarrow (pk,sk)$. The algorithm takes the \blue{$crs$} and a state $aux$ as inputs. It outputs a public/secret key pair $(pk,sk)$.
        \item $\textbf{RegPK}(crs, aux, pk, S_{pk}) \rightarrow mpk, aux'$. The algorithm takes the \blue{$crs$}, a state \blue{$aux$}, and a public key \blue{$pk$} with corresponding attributes $S_{pk} \subseteq \mathcal{U}$ as inputs. It outputs the master public key \blue{$mpk$} and an updated state \blue{$aux'$}.
        \item $\textbf{Encrypt}(mpk, P, m) \rightarrow ct$. The algorithm takes the master public key \blue{$mpk$}, an access policy $P \in \mathcal{P}$, and a message $m$ as inputs. It outputs a ciphertext \blue{$ct$}.
        \item $\textbf{Update}(crs,aux,pk) \rightarrow hsk$. The algorithm takes the \blue{$crs$}, a state \blue{$aux$}, and the public key \blue{$pk$} as inputs. It outputs a helper decryption key \blue{$hsk$}.
        \item $\textbf{Decrypt}(sk, hsk, ct) \rightarrow m \cup \{\perp, \textbf{GetUpdate}\}$. The algorithm takes the secret key \blue{$sk$}, the helper decryption key \blue{$hsk$}, and the ciphertext \blue{$ct$} as inputs. It either outputs a message \blue{$m$}, a decryption failure symbol $\perp$, or a flag \textbf{GetUpdate} that indicates an updated $hsk$ is required.
    \end{itemize}
\end{definition}

\subsection{Sigma Protocol for Discrete Logarithm Equality}
$\mathcal{R}_{DLEQ} = \left\{ (g,G,t,T), s| G=g^s \land T=t^s \right\}$ is the discrete logarithm equality (DLEQ) relation.
A sigma protocol $\prod_{DLEQ}$ for DLEQ is defined as follows.
\begin{itemize}
    \item The prover chooses $r \overset{R}{\leftarrow} \mathbb{Z}_p$ and sends $C_1 = g^r$ and $C_2 = t^r$ to verifier.  
    \item The verifier samples $c \overset{R}{\leftarrow} \mathbb{Z}_p$ and sends it to the prover. 
    \item The prover answers with $z = r+cs$.
    \item The verifier accepts if $g^z = C_1G^c$ and $t^z = C_2T^c.$
\end{itemize}
The property of $\prod_{DLEQ}$ sigma protocol:
\begin{itemize}
    \item \textbf{Completeness.} For any $(\blue{x}, \blue{w}) \in \mathcal{R}$, the prover convinces the verifier to accept with probability 1.
    \item \textbf{Special Soundness.} \blue{Given} any $\blue{x}$ and two accepting conversations on input $\blue{x}$: $(C_1, C_2, c, z)$, $(C_1, C_2, c', z')$, where $c \neq c'$, one can efficiently compute $\blue{w}$ such that $(\blue{x}, \blue{w}) \in \mathcal{R}$.
    \item \textbf{Special Honest-Verifier Zero Knowledge.} There exists a polynomial time simulator that, given $\blue{x}$ and a random challenge $c$, can output an accepting conversation of the form $(C_1, C_2, c, z)$ with the same probability distribution as conversations between the honest prover and verifier on input $\blue{x}$. 
\end{itemize}

The protocol can be made non-interactive using the Fiat-Shamir transformation\cite{fiat1986prove} with a random oracle $H$. $\prod_{DLEQ-NIZK}$ is defined as follows:
\begin{itemize}
    \item $\textbf{Prove}(\blue{x},\blue{w})\rightarrow \pi$. On input $\blue{x}=(g,G,t,T)$ and $\blue{w} = s$, the algorithm samples $r \overset{R}{\leftarrow} \mathbb{Z}_p$ and computes: $$C_1 = g^r, C_2 = t^r, c = H(G, T, C_1,C_2), z = r+cs.$$ Finally, it outputs a proof $\pi = (C_1, C_2, c, z)$.
    \item $\textbf{Verify}(\blue{x},\pi)\rightarrow 0/1$. On input $\blue{x}=(g,G,t,T)$ and proof $\pi = (C_1, C_2, c, z)$, the algorithm checks: $$c \overset{?}{=} H(G, T, C_1, C_2), g^z \overset{?}{=} C_1G^c, t^z \overset{?}{=} C_2T^c.$$ 
    It outputs 1 if all equations hold, indicating that the proof is valid, and 0 otherwise.
\end{itemize}

\subsection{Ethereum Blockchain}
Ethereum\cite{wood2014ethereum} is a decentralized and open-source blockchain platform that incorporates smart contract functionality. \blue{One of the key characteristics of Ethereum is decentralization. Ethereum network is not controlled by a central authority and the decentralized consensus mechanism allows Ethereum to provide a trustless environment, ensuring security and immutability. Smart contracts are self-executing programs with the terms of the agreement directly written into code. They automatically execute the predefined actions without the need for a trusted intermediary. Every transaction or computation on Ethereum requires a certain amount of gas, which serves as a fee to incentive validators. The gas mechanism ensures that resources on the network are used efficiently and prevents abuse by requiring users to pay proportionally to the complexity of their operations.}

\section{System Model and Threat Model}
\label{sec:System Overview}
\subsection{System Model}
Our system consists of seven entities: Key Curator (KC), InterPlanetary File System (IPFS), Blockchain, Data Owners (DOs), Data Users (DUs), Decryption Cloud Server (DCS) Network, and Public Verifiers (PVs). These seven entities are defined as follows:
\begin{itemize}
    \item The KC is responsible for user registration and generates the master public key and auxiliary data for users.
    \item IPFS is a distributed file system used to store ciphertexts in our system without a central server.
    \item Blockchain serves as a decentralized and immutable ledger that records decryption tasks, fraud proofs, and facilitates smart contract execution.
    \item DOs encrypt data with specified access structure and upload the ciphertext to IPFS.
    \item DUs generate their own keys to register with KC. They publish outsourced decryption tasks on blockchain and finish the final decryption locally. Besides, they can generate a fraud proof to accuse dishonest DCS.
    \item The DCS network consists of many servers that compete for tasks. These servers compute and submit task result to blockchain. 
    \item PVs verify the fraud proof and publish the verification result on blockchain.
\end{itemize}

\begin{figure}[h]
    \centering
    {\includegraphics[width=0.45\textwidth]{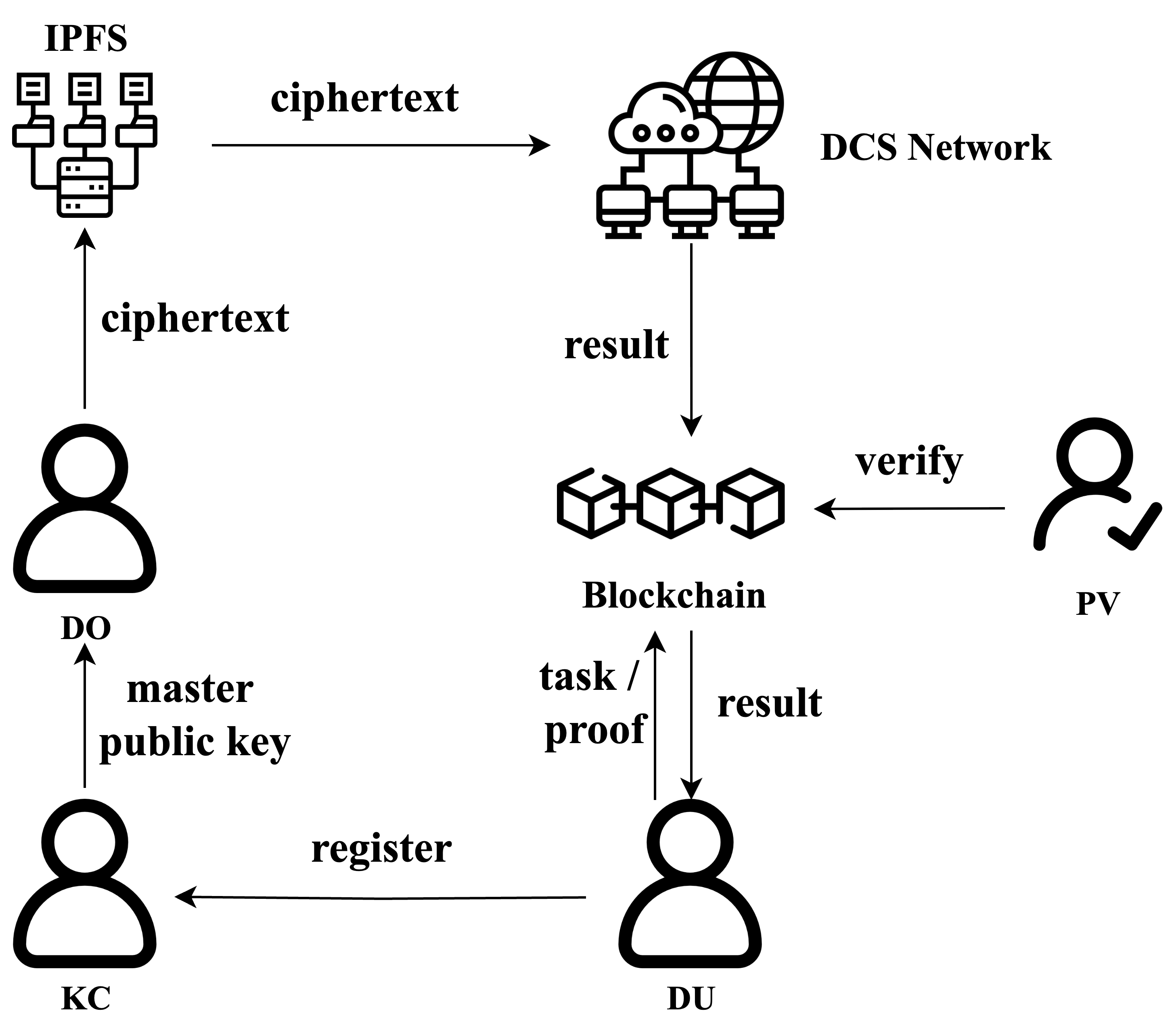}}
    \caption{Brief overview of system entities.}
    \label{fig:model_overview}
\end{figure}

\blue{\Cref{fig:model_overview} illustrates the workflow among all entities in the system model. In brief, the KC stores system public information and is required in user registration. IPFS serves as the data repository to store RABE ciphertext uploaded by the DO. The DU, wishing to decrypt the ciphertext, submits a decryption task on the blockchain and deposits tokens. Any DCS can perform the outsourced decryption and submit the result to the blockchain. In the happy case, the DU retrieves the DO's data and the DCS receives the tokens. In the dispute phase, the deposited tokens are returned to the DU if the fraud proof passes verification by the PV.}

\subsection{Threat Model}
In our system, DOs and IPFS are fully trusted. The KC and PVs are semi-honest and transparent, which means they keep no secrets and perform computations honestly. DCSs are malicious, as they may deliberately submit incorrect results to obtain task rewards. Additionally, since DUs may attempt to submit invalid fraud proof to avoid payment, we also assume they are malicious. The optimistic assumption supposes that all parties are honest unless proven otherwise.

\subsection{Security Requirements}
\begin{itemize}
    \item \textbf{Static Security.} The static security ensures that data is confidential to unauthorized users. 
    
    Let \resizebox{0.95\linewidth}{!}{$\prod_{OSRABE}=(\textbf{Setup},\textbf{KeyGen},\textbf{RegPK},\textbf{Encrypt},\textbf{Transform},$} $\textbf{Decrypt}_{user})$ be a slotted outsourced decryption RABE scheme with security parameter $\lambda$. We define the following game $Game_{sec}$ between a challenger $\mathcal{C}$ and an adversary $\mathcal{A}$:
    \begin{itemize}
        \item \textbf{Setup}: The challenger $\mathcal{C}$ generates \blue{$crs$} and sends it to $\mathcal{A}$. Let the master public key $mpk=\perp$ and the counter \blue{$ctr = 0$}.
        \item \textbf{Query}: The challenger $\mathcal{C}$ responds to the adversary’s key-generation queries on a non-corrupted slot index $i \in N$ with a set of attributes \blue{$S$}. $\mathcal{C}$ computes the counter $ctr = ctr + 1$ and samples $pk_{ctr},sk_{ctr}$ $\leftarrow \textbf{KeyGen}(crs,ctr)$. $\mathcal{C}$ replies to $\mathcal{A}$ with $ctr$ and $pk_{ctr}$.
        Note that in the static security game, the adversary is not allowed to make any corruption queries.
        \item \textbf{Challenge}: The adversary $\mathcal{A}$ selects two messages $\mu_0^*,\mu_1^*$ in message space and an access policy $P^*$ in policy space. The challenger $\mathcal{C}$ calls $\textbf{RegPK}$ to generate $mpk$ and $hsk_i$. Then, $\mathcal{C}$ randomly selects $b\in \{0,1\}$ and replies the challenge ciphertext: $$ct^* \leftarrow \textbf{Encrypt}(mpk, P^*, \mu_b^*).$$ 
        \item \textbf{Output}: $\mathcal{A}$ \blue{outputs} a bit $b' \in \{0,1\}$.
    \end{itemize}
    The advantage of the adversary in $Game_{sec}$ is defined as: 
    
    $\textbf{Pr}[b'= b] - \frac{1}{2}$.
    
    \begin{definition}
    \label{static secure defination}
        An OSRABE scheme is secure if all probabilistic polynomial time (PPT) adversaries have at most a negligible advantage in the above game: $$\textbf{Pr}[b'= b] - \frac{1}{2} \leq negl(\lambda)$$
    \end{definition}
    \item \textbf{Verifiability of outsourced decryption.} The verifiability guarantees that the correctness of transform ciphertext can be efficiently checked by authorized users.
    \item \textbf{Exemptibility of outsourced decryption.}
    The exemptibility guarantees that users cannot falsely accuse the DCS of returning the invalid transform ciphertext.
    \item \textbf{Fairness.} 
    The fairness property in our work guarantees that the DCS gets paid if and only if it has returned a correct result, while the DU must pay for the correct result and is entitled to challenge an incorrect one through fraud proof. Fairness incentivizes each party to perform the protocol honestly.
    \item \textbf{Auditability.} The auditability guarantees the behaviors of the transparent and semi-honest participants (i.e., KC and PV) can be publicly audited.
\end{itemize}

\section{Construction}
\label{sec:Construction}
\subsection{Verifiable Outsourced Decryption Slotted RABE}
\label{subsec: Outsourced Slotted Registered ABE}
A slotted RABE (SRABE) scheme keeps a fixed number of slots and does not support dynamic registrations. Based on scheme \cite{garg2024reducing}, we propose a SRABE with verifiable outsourced decryption (OSRABE) scheme, which is defined as follows:
\begin{itemize}
    \item $\textbf{Setup}(1^{\lambda}, 1^{|\mathcal{U}|}, 1^{L}) \rightarrow crs$. On input the security parameter $\lambda$, the attribute universe $\mathcal{U}$, and the number of slots $L$, the algorithm samples a prime order pairing group: $(\mathbb{G}, \mathbb{G}_T, g, e)$. Let $\mathcal{D} = \{d_i\}_{i \in [1,L]}$ be a progression-free and double-free set. Next, it defines: $$f(i,j) := d_i + d_j \quad \text{and} \quad \mathcal{E} = \{f(i,j)| i,j \in [1,L]:i \neq j\}.$$ For $i \in [1,L]$, it samples $a, b,\gamma_i \overset{R}{\leftarrow} \mathbb{Z}_p$ and sets $t_i = a^{d_i}, \alpha = -a^{3 \cdot d_{max}}$. Then, it computes: $$h = \prod_{i \in [1,L]} g^{a^{3 \cdot  d_{max} - d_i}}, A_i = g^{t_i}, B_i = g^{\alpha}h^{t_i}, P_i = g^{\gamma_i},U_i = g^{bt_i}$$ let $Z = e(g,g)^{\alpha}$ and $W_z = g^{ba^z}$ for each $z\in \mathcal{E}$. Finally, it sets $pp=(\mathbb{G}, \mathbb{G}_T, g, e, Z, h)$ and outputs the common reference string $crs:$ $$crs = (pp,\{(A_i, B_i,P_i,U_i)\}_{i \in [1,L]}, \{W_z\}_{z\in \mathcal{E}}).$$

    \item $\textbf{KeyGen}(crs, i) \rightarrow pk_i, sk_i$. On input the \blue{$crs$} and a slot index $i$, the algorithm first selects $r_i \overset{R}{\leftarrow} \mathbb{Z}_p$ and then computes: $$T_i = g^{r_i}, Q_i = P_i^{r_i}, V_{j,i} = A_j^{r_i}(j \neq i).$$ Finally, it outputs the public key $pk_i$ and the secret key $sk_i:$ $$pk_i = (T_i, Q_i, \{V_{j,i}\}_{j\neq i})\text{ , } sk_i=r_i.$$
    \item $\textbf{RegPK}(crs, \{(pk_i,S_i)\}_{i\in[1,L]})\rightarrow mpk, \{hsk_i\}_{i\in[1,L]}$. On input the \blue{$crs$} and a list of public key-attribute pairs, the algorithm computes $\hat{T}$ and $\hat{V_i}$ for each $i \in [1,L]$: $$\hat{T} = \prod_{j \in [1,L]}T_j \quad, \quad \hat{V}_i = \prod_{j \neq i}V_{i, j}.$$ 
    For each attribute $w \in \mathcal{U}$, it computes $\hat{U_w}$ and $\hat{W_{i,w}}$ for each $i \in [1,L]$: $$\hat{U}_w = \prod_{j \in [1,L]: w \not\in S_j}U_j \quad, \quad \hat{W}_{i,w} = \prod_{j \neq i:w \not\in S_j}W_{f(i, j)}.$$  Finally, it selects three hash functions $H_0, H_1, H_2$ and outputs the master public key \blue{$mpk$} and the helper decryption key $hsk_i$ for each $i \in [1,L]:$
    $$mpk = (pp,\hat{T}, \{\hat{U}_w\}_{w \in \mathcal{U}}, H_0, H_1, H_2),$$ $$\{hsk_i = (i, S_i, A_i, B_i,\hat{V}_i, \{\hat{W}_{i,w}\}_{w \in \mathcal{U}})\}_{i\in[1,L]}.$$
    \item $\textbf{Encrypt}(mpk, (\textbf{M},\rho), m) \rightarrow ct .$ On input the master public key \blue{$mpk$}, the access structure $(\textbf{M},\rho)$ where $\textbf{M}$ is a $\beta \times n$ matrix and $\rho$ maps $\textbf{M}_i$ (i-th row of $\textbf{M}$) to an attribute $\rho(i)$, and a message \blue{$m$}, the algorithm first samples $\mu \in \mathbb{G}_T$ and computes $t_2$ as symmetric key. \textbf{Enc} is a symmetric encryption algorithm. Then, it samples $s,v_2, ...,v_n \overset{R}{\leftarrow} \mathbb{Z}_p$ and $h_1,h_2 \overset{R}{\leftarrow} \mathbb{G}$ such that $h=h_1h_2$ and sets $v = [s,v_2,...,v_n]^T$. For each $k\in[1,\beta]$, it samples $s_k \overset{R}{\leftarrow} \mathbb{Z}_p$ and constructs the ciphertext as follows:
    $$t_1 = H_0(\mu), t_2 = H_1(\mu), C = \textbf{Enc}_{t_2}(m), tag = H_2(t_1||C)$$
    $$C_1 = \mu \cdot Z^s, C_2 = g^s$$ $$C_{3,k} = h_2^{m_k^Tv}\hat{U}_{\rho(k)}^{-s_k}, C_{4,k} = g^{s_k},C_5 = (h_1\hat{T}^{-1})^s$$
    Finally, it outputs the ciphertext $ct:$
    $$ct = ((\textbf{M},\rho), C, C_1, C_2,\{C_{3,k},C_{4,k}\}_{k\in[1,\beta]}, C_5, tag).$$
    \item $\textbf{Transform}(hsk, ct)$  $\rightarrow ct'.$ On input the ciphertext $ct$ and the helper decryption key $hsk$, the algorithm outputs $\perp$ if the set of attributes in $hsk$ does not satisfy the access structure in $ct$. Otherwise, it sets $I=\{j:\rho(j)\in S_i\}$ and computes the set $\{\omega_{j}| \omega_{j} \in \mathbb{Z}_p\}$ such that $\sum_{j\in I} \omega_{j}\textbf{M}_{j} = (1,0,...,0)$.
    Next, it computes:
        \begin{multline*}
        C_1' = \frac{C_1}{e(C_2, B_i)} \cdot e(C_5, A_i)\cdot e(C_2, \hat{V_i}) \cdot \\
        \prod_{j \in I} \big( e(C_{3, j}, A_i) \cdot e(C_{4, j}, \hat{W}_{i, \rho(j)}) \big)^{\omega_{j}}
        \end{multline*}
        $$C_2' = e(C_2, A_i)$$
    Finally, it outputs the transform ciphertext $ct' = (C_1', C_2')$.
    \item $\textbf{Decrypt}_{user}(sk_i, ct', ct)$ $\rightarrow$ $m/\perp.$ On input the original ciphertext $ct$, the transform ciphertext $ct'$, and the secret key $sk_i = r_i$, the algorithm computes:$$\mu' = C_1' \cdot C_2'^{r_i}.$$ Then, it checks: $$ tag \overset{?}{=} H_2(H_0(\mu')||C).$$
    If the equation does not hold, the algorithm outputs $\perp$, \blue{which indicates the transform ciphertext $ct'$ is invalid}. Otherwise, it uses the symmetric decryption algorithm \textbf{Dec} to recover message $m = \textbf{Dec}_{H_1(\mu')}(C)$ and outputs $m$.
\end{itemize}

\subsection{Verifiable Outsourced Decryption RABE}
 Hohenberger et al.\cite{hohenberger2023registered} constructed RABE using multiple SRABE instances. We use the same transformation to construct our ORABE scheme based on the OSRABE scheme described in Section \ref{subsec: Outsourced Slotted Registered ABE}. Our verifiable ORABE is defined as follows:
\begin{itemize}
    \item $\textbf{Setup}(1^{\lambda},1^{|\mathcal{U}|}, 1^{L}) \rightarrow crs$. On input the security parameter $\lambda$, the attribute universe $\mathcal{U}$, and the limit number of users $L=2^l$, the algorithm initializes $l+1$ OSRABE instances and generates $crs_k \leftarrow \textbf{OSRABE.Setup}(1^{\lambda}, 1^{\mathcal{U}}, 1^{2^k})$ for each $k \in [0,l].$ Finally, it outputs: $$crs = (crs_0,...,crs_l).$$
    \item $\textbf{KeyGen}(crs, aux) \rightarrow pk, sk$. On input the $crs$ and the auxiliary data $aux=(ctr,Dict_1,Dict_2,mpk)$, the algorithm generates public/secret key for each OSRABE instance. For each $k \in [0,l]$, it computes slot index $i_k = (ctr\text{ mod } 2^k)+1$ and generates $(pk_k, sk_k) \leftarrow \textbf{OSRABE.KeyGen}(crs_k, i_k).$ Finally, it outputs: $$pk = (ctr,pk_0,...,pk_l),$$ $$sk = (ctr,sk_0,...,sk_l).$$
    
    \item $\textbf{RegPK}(crs,aux,pk,S_{pk}) \rightarrow mpk', aux'$. On input the $crs = (crs_0,...,crs_l)$, the auxiliary data $aux=(ctr_{aux}, Dict_1,Dict_2,$ $mpk)$, where $mpk =(ctr_{aux},mpk_0,...,mpk_l)$, public key $pk = (ctr_{pk},pk_0,...,pk_l)$, and associated attribute set $S_{pk}$, the algorithm executes as follows:
    \begin{itemize}
        \item For each $k \in [0,l]$, it computes $i_k = (ctr_{aux} \text{ mod } 2^k)+1 $ for the k-th scheme. \blue{It terminates if $ctr_{aux} \neq ctr_{pk}$.}
        \item For each $k \in [0,l]$, it sets $Dict_1[k,i_k]=(pk,S_{pk}).$ If $i_k=2^k$, it generates $(mpk'_k, \{hsk'_{k,i}\}_{i\in[1,2^k]})$ {$\leftarrow \textbf{OSRABE.}$} $\textbf{RegPK}(crs_k,Dict_1[k])$ and then updates $Dict_2[ctr+i+1-2^k,k]=hsk'_{k,i}$ for each $i\in[1,2^k].$
    \end{itemize}
    Finally, it outputs: $$mpk' = (ctr_{aux}+1,mpk'_0,...,mpk'_l),$$ $$aux' =(ctr_{aux}+1, Dict_1,Dict_2,mpk').$$
    \item $\textbf{Encrypt}(mpk, (\textbf{M},\rho), m) \rightarrow ct .$ On input the master public key $mpk = (ctr,mpk_0,...,mpk_l)$, the access structure $(\textbf{M},\rho)$ and a message \blue{$m$}, for each $k \in [0,l]$, the algorithm computes: $$ct_k \leftarrow \textbf{OSRABE.Encrypt}(mpk_k,(\textbf{M},\rho), m)$$ If $mpk_k=\perp$, then it sets $ct_k = \perp$. It outputs: $$ct=(ctr,ct_0,...,ct_l).$$
    
    \item $\textbf{Transform}(hsk_k, ct_k)$  $\rightarrow  ct'.$ On input the ciphertext $ct_k$ and the helper decryption key $hsk_k$, the algorithm computes: $$ct' \leftarrow \textbf{OSRABE.Transform}(hsk_k, ct_k).$$ It outputs $ct'$.

    \item $\textbf{Update}(crs,aux,pk) \rightarrow hsk$: On input the common reference string $crs = (crs_0,...,crs_l)$, the auxiliary data $aux=(ctr_{aux}, Dict_1,Dict_2,mpk)$, and the public key $pk=(ctr_{pk},$ 
    
    $pk_0,...,pk_l)$. If $ctr_{pk} <  ctr_{aux}$, it sets $hsk_k=Dict_2[ctr_{pk}+1,k]$ for each $k \in [0,l]$; otherwise, it outputs $\perp$. Finally, it outputs: $$hsk=(hsk_0,...,hsk_l).$$
    
    \item $\textbf{Decrypt}_{user}(sk_k, ct', ct_k)$ $\rightarrow$ $m/\perp.$ On input the transform ciphertext $ct'$, secret key $sk_k$, and original ciphertext $ct_k$, the algorithm \blue{verifies $ct'$ and computes $m$}: $$m/\perp \leftarrow\textbf{OSRABE.}\textbf{Decrypt}_{user}(sk_k, ct', ct_k).$$
    Finally, \blue{it outputs $m$ if the verification of $ct'$ is successful. Otherwise, it outputs $\perp$, which indicates the transform ciphertext $ct'$ is invalid.}
\end{itemize}
\begin{figure*}[h]
    \centering
    \includegraphics[width=\textwidth]{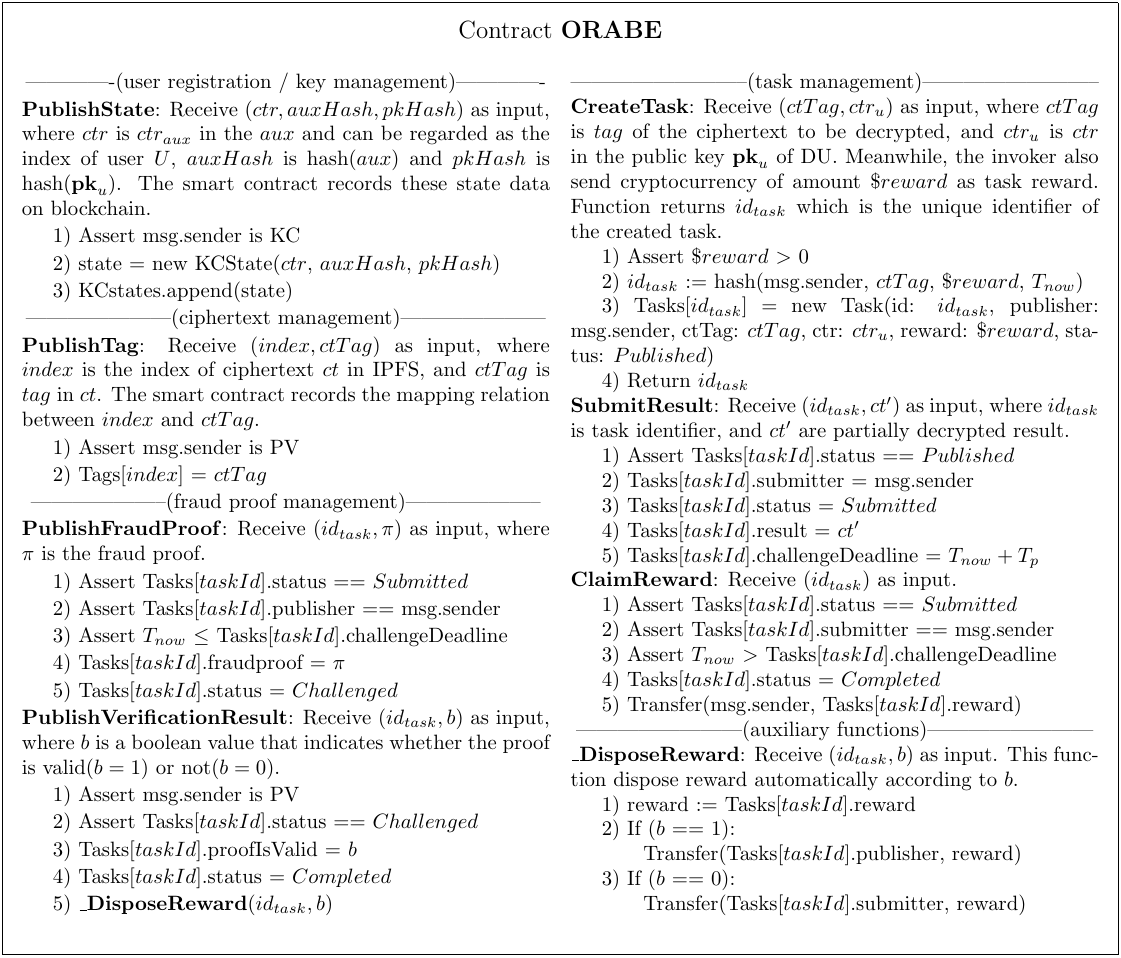}
    \caption{Contract design for ORABE.}
    \label{fig:Contract}
\end{figure*}
\subsection{Exemptibility via NIZK Fraud Proof}
The fraud proof mechanism enables challengers to dispute the previous result by proving its incorrectness. If the decrypt algorithm outputs $\perp$, it indicates that the transform ciphertext is invalid. In this case, the DU could generate a valid fraud proof to accuse the DCS. Our fraud proof is used to prove the transform ciphertext is not correct. This not only provides a mechanism for the DU to demonstrate malicious behavior of the DCS but also ensures that the DCS cannot be maliciously accused, which achieves exemptibility of the outsourced decryption. More specifically, we establish an accountability mechanism based on the following two critical points:
\begin{itemize}
    \item The DU can claim malicious behavior of the DCS through fraud proof, which proves the transform ciphertext with the secret key cannot recover the original plaintext. This mechanism protects the DU from potential economic losses.
    \item It also guarantees exemptibility by preventing false accusations through the NIZK proof that \blue{validates} decryption failures without revealing secret keys. This method effectively resolves the exemptibility vulnerability identified in \cite{ma2015verifiable}, where malicious users could potentially fabricate accusations by revealing invalid private keys. We achieve exemptibility through DLEQ NIZK protocol rather than key exposure, thereby maintaining secret key confidentiality while ensuring fairness.
\end{itemize}

The fraud proof mechanism based on the DLEQ NIZK proof is designed as follows:
\begin{itemize}
    \item $\textbf{Prove}(sk, mpk, ct') \rightarrow \pi$: On input the master public key \blue{$mpk$}, the secret key \blue{$sk$}, and the transform ciphertext $ct'=(C_1',C_2')$, the algorithm computes the verification key $vk = {C'_2}^{sk}$ and generates a NIZK proof $\pi'$ to prove the correctness of $vk$: $$\pi'\leftarrow \prod_{\text{DLEQ-NIZK}}.\textbf{prove}(C'_2, vk, g, g^{sk})$$ Note that $C'_2$ and $pk = g^{sk}$ are public. Finally, it outputs the fraud proof: $$\pi = (\pi', vk).$$
    \item $\textbf{Verify}(\pi, ct', mpk, tag) \rightarrow 0/1$:
    On input the fraud proof $\pi$, the transform ciphertext $ct'$, the master public key \blue{$mpk$}, and the ciphertext tag, the algorithm first checks the correctness of $vk$: $$\prod_{\text{DLEQ-NIZK}}.\textbf{Verify}(\pi', (C'_2, vk, g, T_i))$$ If the proof is invalid, it outputs 0. Otherwise, it computes: $$\mu'' = C_1' \cdot vk.$$ Then, it checks: $$tag \overset{?}{=} H_2(H_0(\mu'')||C).$$ If the equation does not hold, it outputs 1 to indicate the fraud proof is valid (i.e., the transform ciphertext is invalid), and 0 otherwise. 
\end{itemize}

\begin{figure}[h]
    \centering
    \includegraphics[width=0.5\textwidth]{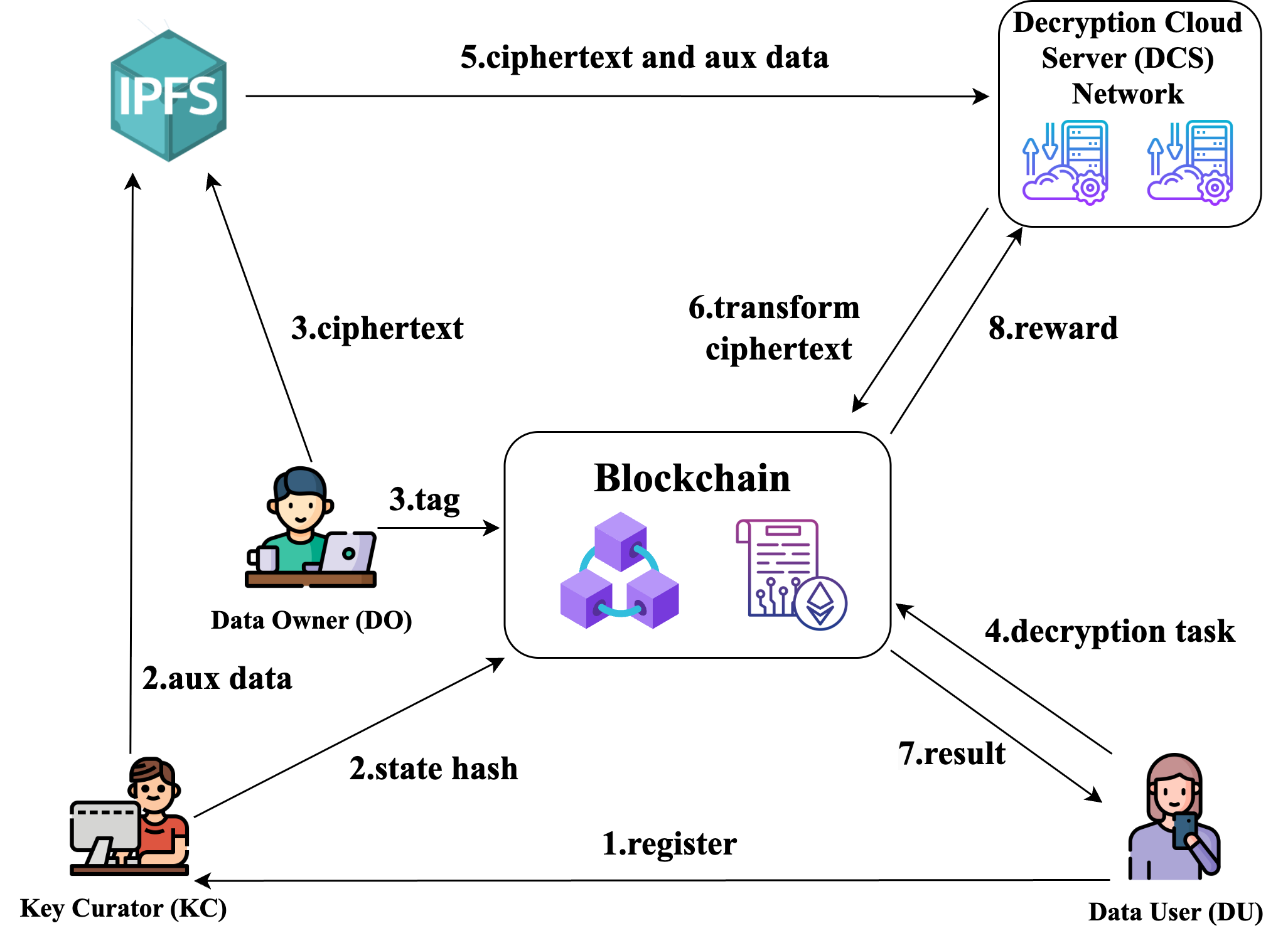}
    \caption{Overview of happy case.}
    \label{fig:happy_overview}
\end{figure}

\noindent\textbf{Discussion.} Unfortunately, this verify algorithm cannot be executed by Ethereum smart contract because it does not support $\mathbb{G}_t$ computation of any elliptic curve now. Therefore, we introduce a semi-honest and transparent public verifier to verify the fraud proof and publish the verification result on blockchain. Ge et al. \cite{ge2023attribute} added precompile contract to native Ethereum to support pairing and $\mathbb{G}_t$ computation on their own blockchain. If we reduce the reliance on native Ethereum security, we can integrate their approach to eliminate the public verifier role and enable the verification of fraud proof through smart contracts directly. In this work, we make a trade-off and prioritize the security of native Ethereum. Note that in our ORABE scheme, the public verifier role can be played by KC, which is semi-honest and transparent.

\subsection{Smart Contract} \label{subsec: Smart Contract ORABE}
The smart contract design is shown in \Cref{fig:Contract}, which manages the whole procedure of user registration, ciphertext tag publication, outsourced decryption, and public verification of fraud proof. Each DU first sends $\mathbf{pk}_u$ and attributes to the KC. On receiving message from each DU, the KC allocates a suitable $ctr$ to the DU and updates the auxiliary data $aux$. To make sure that the KC's behavior is auditable, the KC should invoke \textbf{PublishState} when it finishes registration. The state data consists of $ctr$, hash($\mathbf{pk}_u$), and hash($aux$), which ensures the state transformation can be verified publicly. Ciphertexts are stored in IPFS, and the DO should invoke \textbf{PublishTag} after uploading ciphertext. Regarding outsourced decryption, the DU first publishes an outsourced decryption task by invoking \textbf{CreateTask} with some tokens as the reward. Then, the DCS performs partial decryption and invokes \textbf{SubmitResult} on the blockchain. In the happy case, the DCS could invoke \textbf{ClaimReward} after the challenge window. In the dispute phase, the DU can generate a fraud proof and invoke \textbf{PublishFraudProof} during the challenge window. The PV verifies the fraud proof and invokes \textbf{PublishVerificationResult} on blockchain. 

\subsection{Use Case} \label{subsec: Use case}
We propose a decentralized data sharing framework based on the auditable and reliable ORABE scheme under the optimistic assumption. The optimistic assumption supposes that all parties are honest unless proven otherwise. The dispute phase and public verifier are needed only when one player crashes or attempts to cheat.
\Cref{fig:happy_overview} illustrates the overview of the happy case without challenge, which consists of the following steps:
\begin{enumerate}
    \item The DU generates secret key and public key locally and sends public key and attributes to the KC.
    \item The KC computes and updates the auxiliary data on IPFS, and publishes new state hash and public key hash on blockchain.
    \item The DO encrypts data under a key encapsulation mechanism, uploads ciphertext to IPFS, and publishes ciphertext tag on-chain.
    \item The DU publishes an outsourced decryption task with the reward on-chain. 
    \item The DCS observes the task and downloads the corresponding ciphertext and auxiliary data from IPFS.
    \item The DCS computes and submits the transform ciphertext to blo-
    ckchain.
    \item The DU gets the transform ciphertext from blockchain and completes the decryption locally.
    \item The DCS claims the task reward from blockchain after the challenge window.
\end{enumerate} 

After the final decryption, the DU verifies whether the plaintext matches the ciphertext tag. If the verification fails, as shown in \Cref{fig:dispute_overview}, the DU enters the dispute phase, which consists of the following four steps:
\begin{enumerate}
    \item The DU generates a fraud proof and publishes it on-chain to accuse the DCS within the challenge window.
    \item The PV obtains the fraud proof from blockchain.
    \item The PV verifies the fraud proof and publishes the verification result on blockchain.
    \item The smart contract returns the deposit to the DU if the fraud proof is valid. Otherwise, the DCS receives the payment.
\end{enumerate}

\begin{figure}[h]
    \centering
    \includegraphics[width=0.5\textwidth]{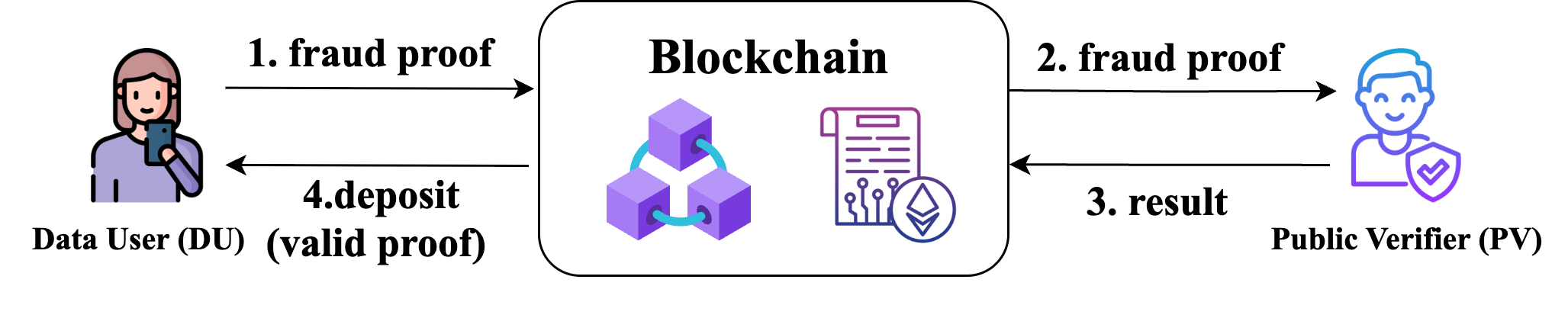}
    \caption{Overview of dispute case.}
    \label{fig:dispute_overview}
\end{figure}

\section{Security Analysis}
\label{sec:Security Analysis}
\begin{theorem}\label{theorem1}\textbf{(OSRABE Static Security)} Suppose the SRABE scheme of \cite{garg2024reducing} is statically secure, then the proposed OSRABE scheme is also statically secure.
\end{theorem}
\begin{proof}\label{prooftheorem1}
    Suppose there exists a PPT adversary $\mathcal{A}$ that can attack our scheme in the CPA security model with advantage $\epsilon$, then we can build a simulator $\mathcal{B}$ that can break the SRABE construction 4.3 in \cite{garg2024reducing} with advantage $\epsilon$.
    
    \textbf{Setup.} The challenger $\mathcal{C}$ generates \blue{$crs$} using the setup algorithm. The simulator $\mathcal{B}$ gets \blue{$crs$} from $\mathcal{C}$ and returns it to the adversary $\mathcal{A}$. Let the master public key $mpk=\perp$ and the counter \blue{$ctr = 0$}. 

    \textbf{Query.} The simulator $\mathcal{B}$ answers the adversary $\mathcal{A}$'s key-generation queries on a non-corrupted slot index $i \in N$ with a set of attributes \blue{$S$} as follows:
    \begin{itemize}
        \item The simulator $\mathcal{B}$ passes the query to the challenger $\mathcal{C}$.
        \item $\mathcal{C}$ computes the counter $ctr = ctr + 1$ and samples $pk_{ctr},sk_{ctr}\leftarrow \textbf{KeyGen}(crs,ctr)$. $\mathcal{C}$ replies to $\mathcal{B}$ with $(ctr, pk_{ctr}).$
        \item $\mathcal{B}$ passes $(ctr, pk_{ctr})$ to $\mathcal{A}$.
    \end{itemize}

    \textbf{Challenge.} $\mathcal{A}$ submits a plaintext pair $(K_0^*, K_1^*)\in \{0,1\}^{2\times k}$ to $\mathcal{B}$, and $\mathcal{B}$ sends them to the challenger $\mathcal{C}$. $\mathcal{C}$ generates the challenge ciphertext pair $(CT_0^*, CT_1^*)$ corresponds to $(K_0^*, K_1^*)$ and sends $CT_b^*(b \in\{0,1\})$ to $\mathcal{B}$.

    \textbf{Guess.} $\mathcal{B}$ passes ciphertext $CT_b^*$ to $\mathcal{A}$ and gets the response bit 0 or 1. $\mathcal{B}$ outputs $\mathcal{A}$'s response as its guess.

    Note that the simulation of $\mathcal{B}$ is perfect if the above game does not abort and the construction 4.3 of \cite{garg2024reducing} has been proven statically secure. The advantage of the simulator breaks the security of the construction in \cite{garg2024reducing} is the same as the adversary that breaks the CPA security of our scheme. Thus, our proposed scheme \ref{subsec: Outsourced Slotted Registered ABE} is statically secure.
\end{proof}

\begin{table*}
\centering
\resizebox{\textwidth}{!}{
\begin{threeparttable}
\caption{Comparison of Pairing-Based RABE Schemes}
    \begin{tabular}{lccccccll}
    \hline
    \textbf{Schemes} & $|\textbf{\textit{crs}}|$ & $|\textbf{\textit{st}}|$ & $|\textbf{\textit{mpk}}|$ & $|\textbf{\textit{hsk}}|$ & $|\textbf{\textit{ct}}|$ & \textbf{Decryption Cost} & \textbf{Assumption} & \textbf{Security} \\
    \hline
    \cite{hohenberger2023registered} & $|\mathcal{U}|L^2$ & $|\mathcal{U}|L^2$ & $|\mathcal{U}|$ & $|\mathcal{U}|$ & $|\mathcal{P}|$ &$|\mathcal{P}|(\mathbb{P} + \mathbb{E}_T)$ & Composite-order & Adaptive \\
    \cite{garg2024reducing} & $L^{(1 + o(1))}$ & $|\mathcal{U}|L$ & $|\mathcal{U}|$ & $|\mathcal{U}|$ & $|\mathcal{P}|$ & $|\mathcal{P}|(\mathbb{P} + \mathbb{E}_T)$ & Prime-Order & Static \\
    \cite{ZZGQ2023registered} & $|\mathcal{U}|L^2$ & $|\mathcal{U}|L^2$ & $|\mathcal{U}|$ & $|\mathcal{U}|$ & $|\mathcal{P}|$ & $|\mathcal{P}|(\mathbb{P} + \mathbb{E}_T)$ & Prime-Order & Adaptive \\
    \cite{attrapadung2024modular} & $L^2$ & $L^2$ & $\delta^2$ & $\delta^2$ & $|\mathcal{P}|$ & $|\mathcal{P}|(\mathbb{P} + \mathbb{E}_T)$ & Prime-Order & Adaptive \\
    \textbf{Ours} & $L^{(1 + o(1))}$ & $|\mathcal{U}|L$ & $|\mathcal{U}|$ & $|\mathcal{U}|$ & $|\mathcal{P}|$ & $\mathbb{E}_T$ & Prime-Order & Static \\
    \hline
    \end{tabular}
    \label{table:RABE comparison}
    \begin{tablenotes} 
        \item Note: \blue{$st$} denotes the auxiliary data, \blue{$L$} denotes the number of users, $\mathcal{U}$ is attribute universe, and $\mathcal{P}$ is access policy. In scheme \cite{attrapadung2024modular}, $\delta = max_{i \in [1,L]}|S_i|$ ($S_i$ is a registered attribute set). $\mathbb{P}$ denotes the pairing operation, $\mathbb{E}_T$ denotes the $\mathbb{G}_T$ exponentiation operation.
    \end{tablenotes} 
\end{threeparttable}
}
\end{table*}

\begin{theorem}\label{theorem2}\textbf{(ORABE Static Security)}
Based on Theorem \ref{theorem1}, the proposed ORABE scheme is statically secure.
\end{theorem}
\begin{proof}\label{prooftheorem2}
    While \cite{hohenberger2023registered} only analyzes the security of the RABE scheme transformed from a fully secure SRABE scheme, the security analysis of the RABE scheme transformed from statically-secure SRABE scheme is proposed in \cite{garg2024reducing}. 
    Therefore, our ORABE scheme inherits the statically secure property from the OSRABE scheme. 
\end{proof}

\begin{theorem}\label{theorem3}
\textbf{(Verifiability)} Assume that $H_0^*$, $H_2^*$ are collision-resistant hash functions, the proposed ORABE scheme achieves the verifiability.
\end{theorem}
\begin{proof}\label{prooftheorem3}
    The verifiability is broken if DCS can return two different transform ciphertexts to match the original tag.
    Then, we can get: $$H_2^*(H_0^*(\mu_1^*)|| C) = tag = H_2^*(H_0^*(\mu_2^*)|| C),\text{where }\mu_1^* \neq \mu_2^*.$$
    If $H_0^*(\mu_1^*) = H_0^*(\mu_2^*)$, then we can break the collision-resistance of $H_0^*$, else we can break the collision-resistance of $H_2^*$.
\end{proof}

\begin{theorem}\label{theorem4}
\textbf{(Exemptibility)} Suppose that the soundness of protocol $\prod_{DLEQ-NIZK}$ holds, the proposed ORABE scheme achieves the exemptibility.
\end{theorem}
\begin{proof}\label{prooftheorem4}
    The exemptibility is broken if DU can return two different valid verification keys $vk_1$ and $vk_2$ in two fraud proofs.
    Then, we can get:
    $$1 \leftarrow \prod_{DLEQ-NIZK}.\textbf{Verify}(\pi_1, (C_2', vk_1, g, T_i))$$ 
    $$1 \leftarrow \prod_{DLEQ-NIZK}.\textbf{Verify}(\pi_2, (C_2', vk_2, g, T_i))$$
    where $vk_1 \neq vk_2.$
    Thus we can break the soundness of $\prod_{DLEQ-NIZK}$ protocol.
\end{proof}

\begin{theorem}\label{theorem5}
\textbf{(Fairness)} The proposed scheme based on Ethereum block-chain ensures fairness.
\end{theorem}
\begin{proof}\label{prooftheorem5}
The fairness of our scheme ensures that the DCS gets paid if and only if it returns a correct result. The DU must deposit tokens into the smart contract to submit the outsourced decryption task. If the DCS submits the correct result, it will receive the reward at the end of the challenge window, which is ensured by exemptibility. Otherwise, the malicious behavior will be detected and proved by the DU, which is ensured by verifiability and NIZK proof, and the DCS get nothing but pay for the transaction gas fee for dishonesty. Since our payments are made through Ethereum smart contracts, we also need to rely on Ethereum consensus protocol. Thus the proposed scheme achieves the fairness property as long as the verifiability, exemptibility, and Ethereum consensus protocol are guaranteed.
\end{proof}

\begin{theorem}\label{theorem6}
\textbf{(Auditability)} The proposed scheme based on publicly verifiable NIZK proof and Ethereum blockchain ensures auditability.
\end{theorem}
\begin{proof}\label{prooftheorem6}
    The auditability ensures the auditing of the malicious behaviors of parties, especially the semi-honest KC and PV in our framework. On the one hand, since our framework is based on decentralized blockchain Ethereum, the transactions and records are transparent and immutable. On the other hand, the hash of registration state related to the KC and the NIZK proof related to the PV are publicly verifiable and stored on-chain. Therefore, the behaviors of the KC and PV are publicly auditable at any time.
\end{proof}

\section{Theoretical Analysis and Evaluation}
\label{sec:Performance And Evaluation}

\subsection{Theoretical Analysis}
\Cref{table:RABE comparison} shows the comparison between some existing RABE schemes and our scheme. Since our scheme is built upon \cite{garg2024reducing}, the parameters are identical to those of \cite{garg2024reducing}, except for the decryption cost. Furthermore, our scheme achieves superior parameter complexity compared to other schemes\cite{hohenberger2023registered, ZZGQ2023registered, attrapadung2024modular}. Additionally, our outsourced decryption significantly reduces the final decryption cost to a constant level, eliminating the need for pairing operations and requiring only a single group exponentiation operation.

\begin{table*}
\centering
\caption{Comparison with OABE Schemes}
\resizebox{\textwidth}{!}{
\begin{tabular}{cccccccc}
\hline
\textbf{Schemes} & \textbf{Authority-Free} & \textbf{Verifiability} & \textbf{Exemptibility} & \textbf{Fairness} & \textbf{Auditability} & \makecell{\textbf{Decryption} \\ \textbf{Cost}} \\  %
\hline
\cite{lai2013attribute} & \XSolidBrush & \Checkmark & \XSolidBrush & \XSolidBrush & \XSolidBrush & 2$\mathbb{E}_G$ + 2$\mathbb{E}_T$ \\  
\cite{li2013securely,qin2015attribute} & \XSolidBrush & \Checkmark & \XSolidBrush & \XSolidBrush & \XSolidBrush & $\mathbb{E}_T$ \\  %
\cite{lin2015revisiting} & \XSolidBrush & \Checkmark & \XSolidBrush & \XSolidBrush & \XSolidBrush & 2$\mathbb{E}_G$ + $\mathbb{E}_T$ \\  
\cite{ma2015verifiable} & \XSolidBrush & \Checkmark & \XSolidBrush & \XSolidBrush & \XSolidBrush & 2$\mathbb{E}_G$ + $\mathbb{E}_T$ \\  %
\cite{ge2023attribute} & \XSolidBrush & \Checkmark & \Checkmark & \Checkmark & \XSolidBrush & $\mathbb{E}_T$\\  %
\textbf{Ours} & \Checkmark & \Checkmark & \Checkmark & \Checkmark &\Checkmark & $\mathbb{E}_T$\\  %
\hline
\end{tabular}
}
\label{table:OABE comparison}
\end{table*}

\Cref{table:OABE comparison} shows the comparison of the previous OABE schemes with our scheme in terms of authority, verifiability, exemptibility, fairness, auditability, and decryption cost. Our scheme is the only one that does not rely on a trusted authority and achieves auditability. Moreover, our scheme achieves the lowest decryption overhead among all the compared schemes. Although all schemes ensure verifiability, only our scheme and \cite{ge2023attribute} guarantee exemptibility and fairness.

\subsection{Experimental Evaluation}
\label{subsec:evaluation}
Our experiment is executed on an Intel(R) Core(TM) i9-12900H CPU and 4 GB RAM running Ubuntu 20.04 LTS 64-bit. The access policy is set as an AND gate of attributes and its size varies from 10 to 100, with a step size of 10.

\noindent \textbf{ORABE Cost.}
We conduct experiments on our scheme and the RABE scheme of \cite{garg2024reducing} and implement them in Python language within the Charm framework. Our benchmark uses SS512 and SS1024 elliptic curves and the SHA-256 hash function. For each algorithm, we repeat 1,000 times to get the average execution time. 

\begin{figure}[!htp]
    \begin{minipage}[c]{0.24\textwidth}
    \centering
    \scalebox{0.5}{\begin{tikzpicture}[thick]
    \begin{axis}[
      xticklabel style={font=\large\bfseries},
      yticklabel style={font=\large\bfseries},
      legend style={at={(0.5,-0.35)},
      legend pos=north west,
      },
      ymajorgrids=true,
      grid style=dashed,
      xlabel={\large\bfseries Number of attributes},
      ylabel={\large\bfseries Time cost (ms)}
      ]
    \addplot[color=red,thick, mark=o]  coordinates{
        (10, 51.23)
        (20, 53.32)
        (30, 76.12)
        (40, 97.17)
        (50, 115.95)
        (60, 133.01)
        (70, 167.92)
        (80, 178.35)
        (90, 206.70)
        (100, 258.58)};
     \addplot[color=blue,thick, mark=o]  coordinates{        
        (10, 378.10)
        (20, 626.76)
        (30, 926.72)
        (40, 1219.73)
        (50, 1528.07)
        (60, 1794.11)
        (70, 2086.85)
        (80, 2403.15)
        (90, 2690.84)
        (100, 2998.74)};     
     \addplot[color=red,thick, mark=x]  coordinates{(10, 11.85) (20, 20.57) (30, 29.86) (40, 39.57) (50, 49.11) (60, 58.38) (70, 67.02) (80, 76.75) (90, 86.22) (100, 95.14)};
     \addplot[color=blue,thick, mark=x]  coordinates{ (10, 305.54) (20, 567.28) (30, 818.26) (40, 1088.62) (50, 1336.43) (60, 1627.02) (70, 1860.09) (80, 2127.15) (90, 2357.32) (100, 2647.87)};     
    \legend{SS512 enc, SS1024 enc, SS512 dec, SS1024 dec}
    \end{axis}
    \end{tikzpicture}}
    \caption{Encrypt and decrypt of \cite{garg2024reducing} \label{fig:sRABE_enc_dec_cost}}
  \end{minipage}
  \begin{minipage}[c]{0.24\textwidth}
    \centering
    \scalebox{0.5}{\begin{tikzpicture}[thick]
    \begin{axis}[
      xticklabel style={font=\large\bfseries},
      yticklabel style={font=\large\bfseries},
      legend style={at={(0.5,-0.35)},
      legend pos=north west,
      },
      ymajorgrids=true,
      grid style=dashed,
      xlabel={\large\bfseries Number of attributes},
      ylabel={\large\bfseries Time cost (ms)},
      ymax = 1.4
      ]
    
     \addplot[color=red,thick, mark=o]  coordinates{(10, 0.06) (20, 0.06) (30, 0.06)  (40, 0.06) (50, 0.06)  (60, 0.06) (70, 0.06) (80, 0.06) (90, 0.06) (100, 0.06)};
     \addplot[color=blue,thick, mark=x]  coordinates{(10, 1.01) (20, 0.98) (30, 1.05) (40, 1.04) (50, 1.03) (60, 1.03) (70, 1.00) (80, 1.01) (90, 1.04) (100, 1.01) };     
    \legend{SS512, SS1024}
    \end{axis}
    \end{tikzpicture}}
    \caption{ORABE $\textbf{Decrypt}_{user}$ \label{fig:OSRABE_dec_cost}}
  \end{minipage}
\end{figure}

As shown in \Cref{fig:sRABE_enc_dec_cost}, the encryption and decryption time cost of the RABE scheme in \cite{garg2024reducing} increases linearly with the number of attributes. \Cref{fig:OSRABE_dec_cost} demonstrates our scheme has very low and constant decryption time cost, which is 196.5$\times$-1585$\times$ \blue{on the} SS512 curve and 304$\times$-2647$\times$ on the SS1024 curve lower than \cite{garg2024reducing} in the number of attributes from 10 to 100. As shown in \Cref{fig:original ciphertext} and \Cref{fig:transform ciphertext}, the size of the original ciphertext increases linearly with the number of attributes, while the transform ciphertext, which contains only two $\mathbb{G}_T$ elements, has constant size, thereby significantly reducing the storage burden on users.

\begin{figure}[!htp]
    \begin{minipage}[c]{0.24\textwidth}
    \centering
    \scalebox{0.5}{\begin{tikzpicture}[thick]
    \begin{axis}[
      xticklabel style={font=\large\bfseries},
      yticklabel style={font=\large\bfseries},
      legend style={at={(0.5,-0.35)},
      legend pos=north west,
      },
      ymajorgrids=true,
      grid style=dashed,
      xlabel={\large\bfseries Number of attributes},
      ylabel={\large\bfseries Size (kb)}
      ]
     \addplot[color=red,thick, mark=o]  coordinates{
        (10, 1.52)
        (20, 3.04)
        (30, 4.56)
        (40, 6.08)
        (50, 7.60)
        (60, 9.12)
        (70, 10.64)
        (80, 12.16)
        (90, 13.68)
        (100, 15.20) };
     \addplot[color=blue,thick, mark=x]  coordinates{        
        (10, 5.88)
        (20, 11.76)
        (30, 17.64)
        (40, 23.52)
        (50, 29.40)
        (60, 35.28)
        (70, 41.16)
        (80, 47.04)
        (90, 52.92)
        (100, 58.80)};     
    \legend{SS512, SS1024}
    \end{axis}
    \end{tikzpicture}}
    \caption{Original ciphertext \label{fig:original ciphertext}}
  \end{minipage}
    \begin{minipage}[c]{0.24\textwidth}
    \centering
    \scalebox{0.5}{\begin{tikzpicture}[thick]
    \begin{axis}[
      xticklabel style={font=\large\bfseries},
      yticklabel style={font=\large\bfseries},
      legend style={at={(0.5,-0.35)},
      legend pos=north west,
      },
      ymajorgrids=true,
      grid style=dashed,
      xlabel={\large\bfseries Number of attributes},
      ylabel={\large\bfseries Size (kb)},
      ymax = 0.65
      ]
    
     \addplot[color=red,thick, mark=o]  coordinates{ (10, 0.25)
        (20, 0.25)
        (30, 0.25)
        (40, 0.25)
        (50, 0.25)
        (60, 0.25)
        (70, 0.25)
        (80, 0.25)
        (90, 0.25)
        (100, 0.25) };
     \addplot[color=blue,thick, mark=x]  coordinates{ (10, 0.52)
        (20, 0.52)
        (30, 0.52)
        (40, 0.52)
        (50, 0.52)
        (60, 0.52)
        (70, 0.52)
        (80, 0.52)
        (90, 0.52)
        (100, 0.52) };     
    \legend{SS512, SS1024}
    \end{axis}
    \end{tikzpicture}}
    \caption{Transform ciphertext \label{fig:transform ciphertext}}
  \end{minipage}
\end{figure}

\noindent \textbf{Fraud Proof Cost.}
The proving and verifying times of the fraud proof are approximately 0.87ms and 1.49ms on the SS512 curve and 11.81ms and 21.40ms on the SS1024 curve, respectively. These results demonstrate that the fraud proof mechanism is efficient for mobile devices, particularly on the SS512 curve.

\noindent \textbf{Gas Usage.}
We implement our ORABE smart contract described in Section~\ref{subsec: Smart Contract ORABE} in solidity. We deploy it on a locally set up Ethereum test network powered by Hardhat and invoke each function 100 times to get the average gas usage. 

\begin{figure}[!htp]
   \begin{minipage}{0.24\textwidth}
    \centering
    \scalebox{0.5}{\begin{tikzpicture}
        \begin{axis}[
            ybar,
            bar width=8pt,
            ylabel={\large\bfseries Gas},
            symbolic x coords={publishTag, publishState, createTask, submitResult, publishFraudProof, publishVerificationResult},
            xticklabel style={rotate=45},
            yticklabel style={font=\large\bfseries},
            ymin=0,
            ymax=1300000,
            legend style={at={(0.5,-0.35)},
              legend pos=north west,
              },
        ]
        \addplot[fill=red!50] coordinates {
            (publishTag, 46878.6) (publishState, 115771.4) (createTask, 255235) (submitResult, 288775.4) (publishFraudProof, 624495.8) (publishVerificationResult, 60204)
        };

        \addplot[fill=blue!50] coordinates {
            (publishTag, 46881) (publishState, 115769) (createTask, 255442.8) (submitResult, 519676) (publishFraudProof, 1199325.6) (publishVerificationResult, 60202.8)
        };
        \legend{SS512, SS1024}
        \end{axis}
    \end{tikzpicture}}
    \caption{Gas of our contract}
    \label{fig:gas_usage_contract}
\end{minipage}
  \begin{minipage}[c]{0.24\textwidth}
    \centering
    \scalebox{0.5}{\begin{tikzpicture}[thick]
            \begin{axis}[
                xlabel={\large\bfseries Number of attributes},
                ylabel={\large\bfseries Gas},
                xticklabel style={font=\large\bfseries},
                yticklabel style={font=\large\bfseries},
                legend style={at={(0.5,-0.35)},
                legend pos=north west,}
            ]
            \addplot[color=red, thick, mark=o] coordinates {
                (10, 600000) (20, 600000) (30, 600000) (40, 600000) (50, 600000)
                (60, 600000) (70, 600000) (80, 600000) (90, 600000) (100, 600000)
            };
            \addplot[color=blue, thick, mark=x] coordinates {
                (10, 1280000) (20, 1280000) (30, 1280000) (40, 1280000) (50, 1280000)
                (60, 1280000) (70, 1280000) (80, 1280000) (90, 1280000) (100, 1280000)
            };
             \addplot[color=orange, thick, mark=triangle] coordinates {
                (10, 8695124) (20, 17028838) (30, 25362373) (40, 33695848) (50, 42029576)
                (60, 50463457) (70, 58845637) (80, 67295726) (90, 75508838) (100, 83695124)
            };
            \legend{Happy case of ours, Challenge case of ours, Scheme [25]}
            \end{axis}
        \end{tikzpicture}}
        \caption{Gas comparison with \cite{ge2023attribute}}
        \label{fig:gas compare [17]}
  \end{minipage}
\end{figure}

As shown in \Cref{fig:gas_usage_contract}, all functions in our smart contract exhibit constant gas usage. \Cref{fig:gas compare [17]} presents a comparison of gas usage across three scenarios: the happy case and challenge case of our scheme, and the scheme proposed by \cite{ge2023attribute}. Note that the gas usage of pairing operations in \cite{ge2023attribute} cannot be tested on native Ethereum due to incompatibility. Therefore, we adopt the experimental result $168164 + k \times 46308$ reported in their work, where $k$ is the number of pairing operations. The results demonstrate that our scheme maintains constant gas usage and achieves significantly lower usage compared to \cite{ge2023attribute}.

\section{Conclusion}
\label{sec:Conclusion}
In this paper, we propose ORABE, an auditable registered attribute-based encryption scheme with reliable outsourced decryption based on blockchain. Our scheme achieves verifiability through a verifiable tag mechanism and ensures exemptibility via NIZK proofs. Building upon ORABE, we design a decentralized and auditable data sharing framework without a centralized authority. Additionally, we provide a formal security analysis of the proposed scheme. Finally, we compare the theoretical complexity and functionality of our scheme with existing works and evaluate it on Ethereum to demonstrate feasibility and efficiency. In future work, we plan to explore the applications of our proposed scheme in various innovative scenarios within the metaverse, particularly in data sharing of emerging domains such as SocialFi (Social Finance) and GameFi (Game Finance).

\fontsize{8.5pt}{10pt}\selectfont

\section*{Acknowledgement}
This work was supported by National Natural Science Foundation of China (Grant Nos. 62272107 and 62302129), the Project of Shanghai Science and Technology (Grant No. 2024ZX01), Hainan Province Key R\&D plan project (No. ZDYF2024GXJS030), the HK RGC General Research Fund (Nos. 16208120 and 16214121), and the Project of the Yiwu Research Institute of Fudan University.

\section*{Competing Interest}
The authors declare that they have no competing interests or financial conflicts to disclose.

\bibliographystyle{fcs}
\bibliography{main}

\end{document}